\documentclass{sig-alternate}

\usepackage{subfigure}
\usepackage{amsmath}
\usepackage{amsfonts}
\usepackage{amssymb}
\usepackage{graphicx}
\usepackage{hyperref}
\usepackage{graphicx}
\usepackage{xspace}
\usepackage{multirow}
\usepackage{enumitem}
\usepackage{here}
\usepackage{stfloats}

\newcommand{\M}{\ensuremath{\mathcal{M}}\xspace}
\newcommand{\R}{\ensuremath{\mathcal{R}}\xspace}

\newcommand{\theadvisor}{the{\bf advisor}\xspace}
\newcommand{\RWR}{\textsc{PaperRank}\xspace}
\newcommand{\WRWR}{\textsc{DaRWR}\xspace}
\newcommand{\WKatz}{\textsc{DaKatz}\xspace}
\newcommand{\Katz}{\textsc{Katz}\xspace}

\newtheorem{definition}{Definition}
\newtheorem{theorem}{Theorem}

\newcommand{\union}{\cup}

\makeatletter
\renewcommand*{\thetable}{\arabic{table}}
\renewcommand*{\thefigure}{\arabic{figure}}
\let\c@table\c@figure
\makeatother

\begin{document}

\title{Recommendation on Academic Networks using\\ Direction Aware Citation Analysis}

\numberofauthors{1}
\author{
Onur K\"{u}\c{c}\"{u}ktun\c{c}$^{1,2}$, Erik Saule$^1$, Kamer Kaya$^1$, \"{U}mit V. \c{C}ataly\"{u}rek$^{1,3}$\\
\affaddr{$^1$ Dept. Biomedical Informatics, The Ohio State University}\\
\affaddr{$^2$ Dept. Computer Science and
Engineering, The Ohio State University}\\
\affaddr{$^3$ Dept. Electrical and Computer Engineering, The Ohio State University}\\
\email{\{kucuktunc,kamer,esaule,umit\}@bmi.osu.edu}\\
}

\sloppy
\maketitle

\begin{abstract}
The literature search has always been an important part of an academic
research. It greatly helps to improve the quality of the research
process and output, and increase the efficiency of the researchers in
terms of their novel contribution to science. As the number of published 
papers increases every year, a manual search becomes more exhaustive even 
with the help of today's search engines since they are not specialized for 
this task. In academics, two relevant papers do not always have to share 
keywords, cite one another, or even be in the same field. Although a 
well-known paper is usually an easy pray in such a hunt, relevant papers 
using a different terminology, especially recent ones, are not obvious to 
the eye. 

In this work, we propose paper recommendation algorithms by using the
citation information among papers. The proposed algorithms are {\em
direction aware} in the sense that they can be tuned to find either
recent or traditional papers. The algorithms require a set of papers
as input and recommend a set of related ones. If the user wants to
give negative or positive feedback on the suggested paper set, the
recommendation is refined.  The search process can be easily guided in
that sense by relevance feedback.  We show that this slight guidance
helps the user to reach a desired paper in a more efficient way. We
adapt our models and algorithms also for the venue and reviewer
recommendation tasks. Accuracy of the models and algorithms is
thoroughly evaluated by comparison with multiple baselines and
algorithms from the literature in terms of several objectives specific
to citation, venue, and reviewer recommendation tasks. All of these
algorithms are implemented within a publicly available web-service
framework which currently uses the data from
DBLP\footnote{\url{http://dblp.uni-trier.de}} and
CiteSeer\footnote{\url{http://citeseer.ist.psu.edu/}} to construct the
proposed citation graph.
\end{abstract}


\section{Introduction}
\label{sec:intro}

The academic community has published millions of research papers to
date and the number of new papers has been increasing with time. For
example, based on DBLP, computer scientists published 3 times more
papers in 2010 than in 2000~(see
Figure~\ref{fig:paper_count}-left). With more than one hundred
thousand new papers each year, performing a complete literature search
became a herculean task. A paper cites in average 20 other papers (see
Figure~\ref{fig:paper_count}-right), which means that there might be
more than a thousand papers that cite or are cited by any paper a
researcher write. Researchers typically rely on manual methods to
discover new research such as keyword-based search on search engines,
reading proceedings of conferences, browsing publication list of known
experts or checking the reference list of paper they are interested. 
These techniques are time-consuming and only allow to reach a limited
set of documents in a reasonable time. Developing tools that help
researchers find unknown and relevant papers will certainly
increase the productivity of the scientific community.

Some of the existing approaches and tools for the literature search
cannot compete with the size of today's literature. Keyword-based
approaches suffer from the confusion induced by different names of
identical concepts in different fields. (For instance, {\em partially
ordered set} or {\em poset} are also often called {\em directed
acyclic graph} or {\em DAG}). Hence, a researcher may not be
able to find the right paper even she is suggested to scan a long
list of papers by a keyword-based approach. Conversely, two different
concepts may have the same name in different fields (for instance,
{\em hybrid} is commonly used to specify software hybridization,
hardware hybridization or algorithmic hybridization) and such homonyms
may drastically increase the number of suggested but unrelated papers.
Some publishers and digital libraries automatically suggest papers to
authors; however, their suggestions are usually based on the
publication history of the researcher which may not match with her
current interests.

\begin{figure}[t]
\begin{center}
\includegraphics[width=0.51\columnwidth,page=1]{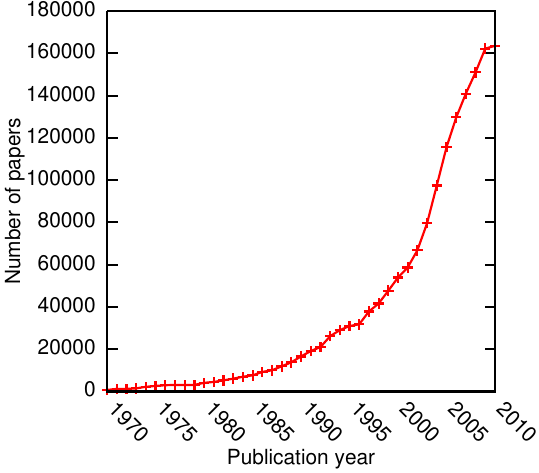}
\includegraphics[width=0.47\columnwidth,page=2]{fig1.pdf}
\vspace{-2em}
\caption{Number of new papers published each year based on DBLP (left), and 
number of papers with given citation and reference count (right).}
\label{fig:paper_count}
\end{center}
\end{figure}

To achieve this goal, we built a publicly available web service called 
\theadvisor\footnote{\url{http://theadvisor.osu.edu/}}. It takes a 
bibliography file
containing a set of papers, i.e., {\it seeds}, as an input to initiate
the search. The user can specify that she is interested in classical
papers or in recent papers. Then, the service returns a set of
suggested papers ordered with respect to a ranking function. The user
can guide the search or {\it prune} the list of suggested papers with
a positive or negative feedback by declaring a subset {\em relevant}
or {\em irrelevant}. In this case, the service completely refines the
set and shows the new results back to the researcher. In addition to
papers, the service also suggests researchers or experts, and
conferences or journals of interest. We believe that it will be a
valuable asset of a researcher while performing several tasks, such
as:
\begin{itemize}[itemsep=2pt,parsep=2pt,leftmargin=1.2em]
\item searching the literature in any topic she is interested, 
\item finding recent or traditional papers related to a problem,
\item improving the reference list of a manuscript being written, 
\item finding conferences and journals for attendance, subscription,
  or paper submission,
\item finding a set of researchers in a field of interest to follow their work, 
\item finding a list of potential reviewers, which is required by
  certain journals in the submission process.
\end{itemize}

The service uses the bibliographical information while suggesting
relevant papers, venues, and people to the researcher. For each paper, it uses the
authorship and venue information in addition to the list of papers it
cites. The service works on a modified version of the {\it citation
graph} which is constructed by using this information. In other
words, the service recommends papers, experts, and venues using citation
analysis. We do not take the textual data into account because our
aim is finding all conceptually related and high quality documents
even they use a different terminology. It has been shown that
text-based similarity is not sufficient for this task and that most of
the relevant informations are contained within the citation
graph~\cite{Strohman07}. Besides, it is plausible that there is
already a correlation between citation similarities and text
similarities of the papers~\cite{Salton63}.

Our aim in this work is to evaluate the existing algorithms and to
explain the new algorithms that power our service. We distinguish two
types of algorithms in the literature. Some algorithms (such as
Cocitation~\cite{Small73}, Cocoupling~\cite{Kessler63} and
CCIDF~\cite{Lawrence99}) only use direct citations and references of
the seed papers. Other methods (such as PaperRank~\cite{Gori06} and
Katz~\cite{Strohman07}) perform a deep search of the citation graph by
traversing all its edges; they are often said to be eigenvector
based. However, none of these methods allow explicitly to search the
paper space looking for either old or recent papers.

In this work, we present the class of {\em direction aware}
algorithms. They feature a parameter which allows to give more
importance to either the citation of papers or their references. This
parameter makes the citation suggestion process easily tunable
for finding either recent or traditional relevant papers. In
particular we extend two eigenvector based methods into direction aware
algorithms, namely~\WRWR and \WKatz.

This paper presents an evaluation of the existing and proposed
algorithms for citation recommendation under the light of link
prediction and citation patterns. We also investigate the potential
of the positive and negative feedback mechanism our service
exposes. Finally we show that citation recommendation can be used to
recommend venues and reviewers better than methods commonly used by
researchers.

The paper is organized as follows: In Section~\ref{sec:rw}, we briefly
present a survey for related work. The problems and the methods are
formally presented in Section~\ref{sec:method}. The accuracy of the
methods is experimentally analyzed in Section~\ref{sec:exp}. 
Section~\ref{sec:conclusion} discusses about future work and concludes 
the paper.

\section{Related work}
\label{sec:rw}
Citation analysis has been successfully used for various tasks
including expert finding~\cite{Bogers08}, academic evaluation of
researchers, conferences, journals and
papers~\cite{Garfield79,Hirsch05}, context-aware citation
recommendation~\cite{He10}, and impact prediction~\cite{shi10}.

There are various citation analysis-based paper recommendation methods
depending on a pairwise similarity measure between two papers. Bibliographic
coupling, which is one of the earliest works, considers papers having
similar citations as related~\cite{Kessler63}. Another early work, the
Cocitation method, considers papers which are cited by the same
papers as related~\cite{Small73}. A similar cites/cited approach by using
collaboration filtering is proposed by McNee~et~al.~\cite{McNee02}.
Another method, common citation $\times$ inverse document frequency~(CCIDF)
also considers only common citations, but by weighting them
with respect to their inverse frequencies~\cite{Lawrence99}. 

More recent works define different measures such as Katz which is
proposed by Liben-Nowell and Kleinberg for a study on the link
prediction problem on social networks~\cite{Liben-Nowell03} and used
later for information retrieval purposes including citation
recommendation by Strohman~et~al.~\cite{Strohman07}. For two papers in
the citation network, the Katz measure counts the number of paths by
favoring the shorter ones. Lu~et~al. stated that both bibliographic
coupling and Cocitation methods are only suitable for special cases
due to their very local nature~\cite{Lu06}. They proposed a method
which computes the similarity of two papers by using a vector based
representation of their neighborhoods in the citation network and
compared the method with CCIDF. Liang~et~al. argued that most of the
methods stated above considers only direct references and citations
alone~\cite{Liang11}. Even Katz and the vector based method
of~\cite{Lu06} consider the links in the citation network as simple
links. Instead, Liang~et~al. added a weight attribute to each link and
proposed the method Global Relation Strength which computes the
similarity of two papers by using a Katz-like approach.

Many works use random walk with restarts~(RWR) for citation
analysis~\cite{Gori06,Ma08,Li09,Lao10}. RWR is a well known and
efficient technique used for different tasks including computing the
relevance of two vertices in a graph~\cite{Pan04}. It is very similar
to the well known PageRank algorithm which is used by Both Li and
Willett~\cite{Li09}~(ArticleRank) and Ma~et~al.~\cite{Ma08} to
evaluate the importance of the academic papers. Gori and
Pucci~\cite{Gori06} proposed an algorithm PaperRank for RWR-based
paper recommendation which can also be seen as a Personalized PageRank
computation~\cite{Jeh2003} on the citation graph. Lao and
Cohen~\cite{Lao10} also used RWR for paper recommendation in citation
networks and proposed a learnable proximity measure for weighting the
edges by using machine learning techniques.

As far as we know, none of these works study the recent/traditional
paper recommendation problem. The closest work is Claper~\cite{Wang10}
which is an automatic system that measure how much a paper is
classical, allowing to rank a list of paper to highlight the most
classical ones.

\section{Problems and methods}
\label{sec:method}
Let $G=(V,E)$ be the \emph{citation graph}, with $n$ papers
$V=\{v_1,\ldots,v_n\}$. In $G$, each directed edge $e=(v_i,v_j)\in E$
represents a \emph{citation} from $v_i$ to $v_j$. For the rest of the
paper, we use the phrases \emph{``references of $v$''} and
\emph{``citations to $v$''} as to describe the graph around vertex $v$
(see Figure~\ref{fig1}). We use $deg^-(v)$ and $deg^+(v)$ to denote
the number of references of and citations to $v$, respectively. 

\begin{figure}[h]
\begin{center}
\includegraphics[width=0.8\columnwidth]{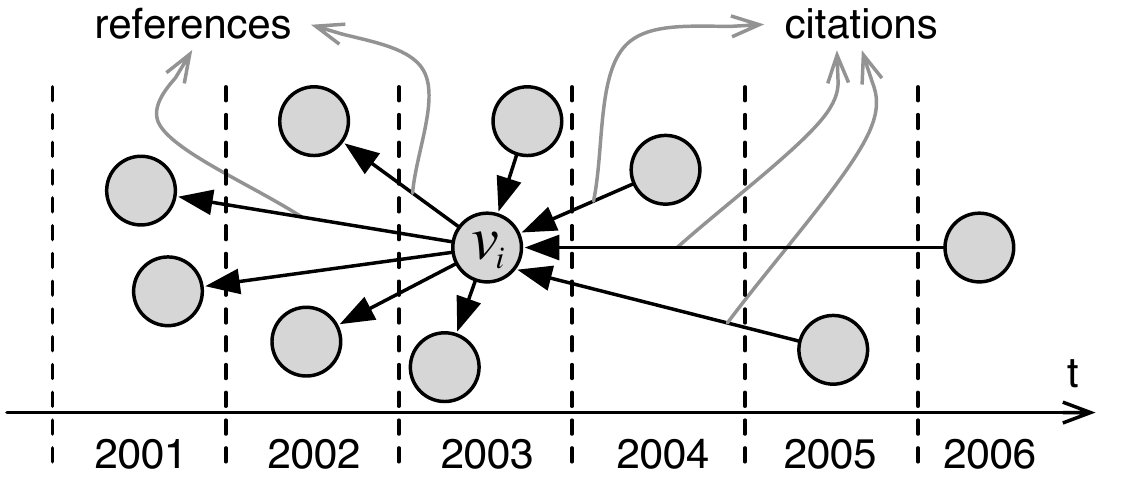}
\caption{Citation graph around a paper $v_i$ with references and 
citing papers.}
\label{fig1}
\end{center}
\end{figure}

In this work, we consider three query types:
\begin{itemize}
\item {\bf Paper recommendation (PR):} Given a set of $m$ seed papers $\M
= \{p_1,\ldots,p_m\}$ and a parameter $k$ s.t. $\M \subset V$, return
top-$k$ papers which are relevant to the ones in $\M$. 

\item {\bf Venue recommendation (VR):} Given a set of $m$ seed papers 
$\M=\{p_1,\ldots,p_m\}$ and a parameter $k$, return top-$k$ venues
related to the papers in $\M$.

\item {\bf Expert recommendation (ER):} Given a set of $m$ seed papers 
$\M=\{p_1,\ldots,p_m\}$ and parameter $k$, return top-$k$ experts
studying on topics related to the papers in $\M$.
\end{itemize}

These query definitions are generic. They can be used for various
academic tasks by the researchers. In this paper, we target
the manuscript preparation and submission process since all
of queries above are useful in this process: executing a PR query is
a very efficient way of finding overlooked citations in a manuscript with
the cited papers as the input $\M$. VR queries are useful
while deciding the conference or journal for submission. And ER queries are useful while submitting a manuscript to some journals which require a set of names of potential reviewers.

\subsection{Citation recommendation}
\label{sec:cr}

\subsubsection{Random walk with restart}

\RWR is based on random walks in the citation graph $G$. The current
structure of $G$ is not suitable for finding recent and relevant
papers since such papers have only a few incoming edges. Moreover, since the
graph is acyclic, all random walks will end up on old papers. To
alleviate this, given a PR query with inputs $\M$ and $k$, \RWR
constructs a directed graph $G'=(V',E')$ by slightly modifying the
citation graph $G$ as follows:

\begin{itemize}[itemsep=2pt,parsep=2pt]
\item A source node $s$ is added to the vertex set:
\begin{align*}
V'&=V\cup\{s\}
\end{align*}
\item Back-reference edges~($E_{b}$), the edges from $s$ to
seed papers~($E_{f}$), and restart edges from $V$ to $s$~($E_r$) are added to the graph:
\begin{align*}
E_b&=\{(y,x):(x,y)\in E\}\\
E_f&=\{(s,v):v\in\M\}\\
E_r&=\{(v,s):v\in V\}\\
E'&=E\cup E_b\cup E_f\cup E_r
\end{align*}
\end{itemize}

\begin{figure}[hb]
\begin{center}
\includegraphics[width=0.7\columnwidth]{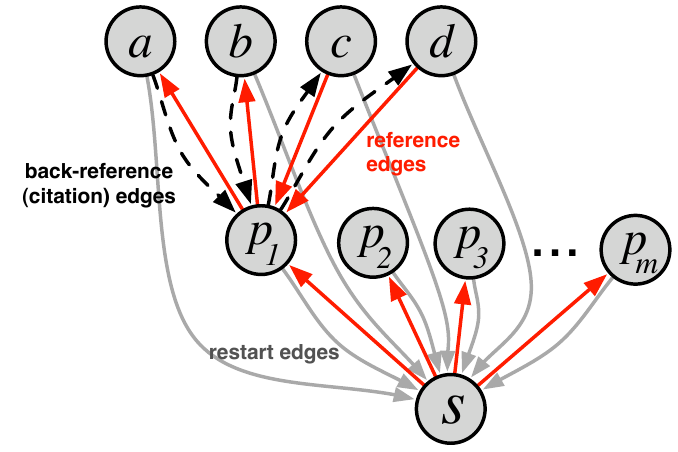}
\caption{Citation graph with source node $s$ and
  seed set $\M=\{p_1,\ldots,p_m\}$. The papers $a$ and $b$ are cited by $p_1$,
  where $c$ and $d$ cites $p_1$. Note that there is a corresponding
  back-reference edge for every reference.}
\label{fig2}
\end{center}
\end{figure}

The new directed graph $G'$ has \emph{reference}~(red),
\emph{back-reference}~(dashed), and \emph{restart}~(gray) edges~(see
Figure~\ref{fig2}). In this model, the random walks are directed
towards both references and citations of the papers. In addition, the
restarts from the source vertex $s$ will be distributed to only the
seed papers in $\M$. Hence, random jumps to any paper in the
literature are prevented. We assume that a random walk ends in $v$
continues with a neighbor with a damping factor $d \in (0,1]$. And with probability
$(1-d)$, it restarts and goes to the source $s$.  Let $R_{t-1}(v)$ be the
probability of a random walk ends at vertex $v \neq s$ at iteration
$t-1$. Let $C_t(v)$ be the contribution of $v$ to one of its neighbors
at iteration $t$. In each iteration, $d$ of $R_{t-1}(v)$ is
distributed 
among its references and citations equally. Hence,
\begin{equation}
C_t(v) = d \frac{R_{t-1}(v)}{deg^+(v) + deg^-(v)}.\label{eq:contrib}
\end{equation}

Initially, a probability score of $1$ is given to the source node,
meaning that a researcher expands the bibliography starting with the
paper itself:
\begin{equation}
R_0(x)=\begin{cases}
1, & \text{if $x=s$}\\
0, & \text{otherwise}
\end{cases}
\label{eqn:prinit}
\end{equation}
where $R_0$ is the probability at $t = 0$. The $\RWR$ algorithm
computes the probability of a vertex $u$ at iteration $t$ as
\begin{equation}
R_t(u)=\begin{cases}
(1-d)\sum_{v \in V} R_{t-1}(v), & \text{if $u=s$}\\
\sum_{(u,v)\in E}{C_{t}(v)} + \frac{R_{t-1}(s)}{|{\M}|}  , &
\text{if $u \in \M$}\\
\sum_{(u,v)\in E}{C_{t}(v)}, & \text{otherwise.}
\end{cases}
\label{eqn:pr}
\end{equation}

The $\RWR$ algorithm converges when the probability of the papers are
stable, i.e., when the process is in a {\it steady state}. Let
$$\Delta_{t} = (R_t(u_1)-R_{t-1}(u_1), \ldots,R_t(u_n)-R_{t-1}(u_n))$$
be the difference vector. We say that the process is in the steady
state when the L2 norm of $\|\Delta_t\|$ is smaller than given value $\epsilon$. That is, $$\|\Delta_{t}\| = \sqrt{\sum_{u\in V}
\left(R_t(u)-R_{t-1}(v)\right)^2} < \epsilon.$$ For a given set of
initial papers $\M$, and parameters $d$ and $\epsilon$, suppose the
algorithm converges.
\begin{definition}
  The {\it relevance score} of a paper $u$ with respect to the seed
  papers is equal to the steady state probability $R(u)$.
\end{definition}
We choose the top-$k$ non-seed papers with the highest relevance
scores as the initial recommended paper set $\R_{paper}$.

\begin{theorem}
The $\RWR$ algorithm converges to a steady state in a finite number
iterations. Furthermore, there is only one steady state distribution
and hence, the relevance scores are unique.
\end{theorem}

\begin{proof}
Consider the subgraph $H = (V_H, E_H) \subseteq G'$ induced by the
source $s$ and all vertices reachable from the source. That is, $V_H =
\{u \in V': R_t(u) > 0\}$ and $E_H = (V_H \times V_H) \cap E'$. For each
$u \in V_H \setminus \{s\}$ there is a directed edge $(u,s)$ and a
directed path $s \rightarrow u$. Hence, each vertex pair in $V_H$ is
connected to each other and $H$ is strongly connected. Thus, the
transition matrix of the corresponding Markov chain is irreducible. 
Hence, the steady state exists and is unique. \qed
\end{proof}

\subsubsection{Direction aware random walk with restart}

A random walk with restart is a good way to find relevance scores of
the papers. However, the $\RWR$ algorithm treats the citations and
references in the same way. This may not lead the researcher to recent
and relevant papers if she is more interested with those. Old and well
cited papers have an advantage with respect to the relevance scores
since they usually have more edges in $G'$. Hence $G'$ tends to have
more and shorter paths from the seed papers to old papers. We define a
{\it direction awareness} parameter $\lambda \in [0,1]$ to obtain more
recent results in the top-$k$ documents. We then define two types of
contributions of each paper $v$ to a neighbor paper in iteration $t$:
\begin{align}
C^{+}_t(v) &= d \lambda  \frac{R_{t-1}(v)}{deg^+(v)}, \\
C^{-}_t(v) &= d (1-\lambda) \frac{R_{t-1}(v)}{deg^-(v)},
\end{align}
where $C^{-}_t(v)$ is the contribution of $v$ to a paper in its
reference list and $C^{+}_t(v)$ is the contribution of $v$ to a paper
which cites $v$. Hence, for a non-seed, non-source paper $u$, 
\begin{equation}
\hspace*{-1ex}R_t(u) = \sum_{(v,u)\in E_b}C^{+}_t(v)
+  \sum_{(v,u) \in E} C^{-}_t(v).
\label{eqn:pr2}
\end{equation}
For a seed node $u$, the $R_t(u)$ is computed similarly except that each 
seed node has an additional
$\frac{R_{t-1}(s)}{|{\M}|}$ in the equation. $R_t(s)$ is computed in
the same way as~\eqref{eqn:pr}. With this modification, the parameter
$\lambda$ can be used to give more importance either to traditional
papers with $\lambda \in [0,0.5]$ or recent papers with
$\lambda \in [0.5, 1]$. We call this
algorithm {\it direction aware random walk with
restart}~($\WRWR$).

Note that $\WRWR$~\eqref{eqn:pr2} has the \emph{probability leak}
problem when a paper has no references or citations. If this is the
case some part of its score will be lost at each iteration. For such
papers, we distribute the whole score from the previous iteration
towards only its references or citations.

\subsubsection{Katz and direction awareness}
\label{sec:da}
The direction awareness can be also adapted to other similarity
measures such as the graph-based Katz distance
measure~\cite{Liben-Nowell03} which was used before for the citation
recommendation purposes~\cite{Strohman07}. With Katz measure, the
similarity score between two papers $u,v \in V$ is computed as
$$Katz(u,v) = \sum_{i = 1}^L \beta^i |paths^i_{u,v}|,$$ where
$\beta\in [0,1]$ is the decay parameter, $L$ is an integer parameter,
and $|paths^i_{u,v}|$ is the number of paths with length $i$ between
$u$ and $v$ in the graph with paper and back-reference edges $G''=(V,
E \cup E_b)$. Notice that the path does not need to be elementary,
i.e., the path $uvuv$ is a valid path of length 3. Therefore the Katz
measure might not converge for all values of $\beta$ when $L =
\infty$. $\beta$ needs to be chosen smaller than the larger eigenvalue
of the adjacency matrix of $G''$. And in practice $L$ is set to a
fixed value (in our experiment $L = 10$).
In our context with multiple seed papers, the relevance of a paper $v$
is set to $R(v) = \sum_{u\in\M} Katz(u,v)$.\\

We extend the Katz distance by using direction awareness to weight the
contributions to references and citations differently with the
$\lambda$ parameter as in \WRWR:
$$DaKatz(u,v) = \sum_{i=1}^L \left[ \lambda\beta^i |Rpaths^i_{u,v}| +
  (1-\lambda)\beta^i |Cpaths^i_{u,v}| \right],$$ where
$|Rpaths^i_{u,v}|$ (respectively, $|Cpaths^i_{u,v}|$) is the number of
paths in which the last edge in the path is a reference edge of $E$
(respectively, a citation edge of $E_b$).

\subsection{Venue and Reviewer recommendation}
\label{sec:vr}

Given a VR query with inputs $\M$ and $k$, we execute the paper
recommendation process and obtain the relevance scores of all papers
in the database. The relevance score of each venue $\nu$ is computed
as the sum of relevance scores of all papers published in that venue,
i.e.,
$$R(\nu) = \sum_{\text{$u$ is published in $\nu$}}{R(u)}.$$
We then choose the top-$k$ venues with the highest relevance
scores as the suggestion set $\R_{venue}$.

Similarly, given an ER query with inputs $\M$ and $k$, we execute the paper
recommendation process and obtain the relevance scores of all papers in the
database. The relevance score of each expert $\alpha$ is computed as
the sum of relevance scores of all papers written by $\alpha$, i.e.,
$$R(\alpha) = \sum_{\text{$u$ is written by $\alpha$}}{R(u)}.$$ We
then choose the top-$k$ researchers with the highest relevance scores
as the suggestion set $\R_{expert}$.

\section{Experiments}
\label{sec:exp}

We carefully evaluate the accuracy of the proposed direction aware
algorithms by comparing them with existing baselines and
algorithms. Here, we give the details and results of these
experiments.

\begin{figure}[t]
\centering
\includegraphics[width=.9\linewidth]{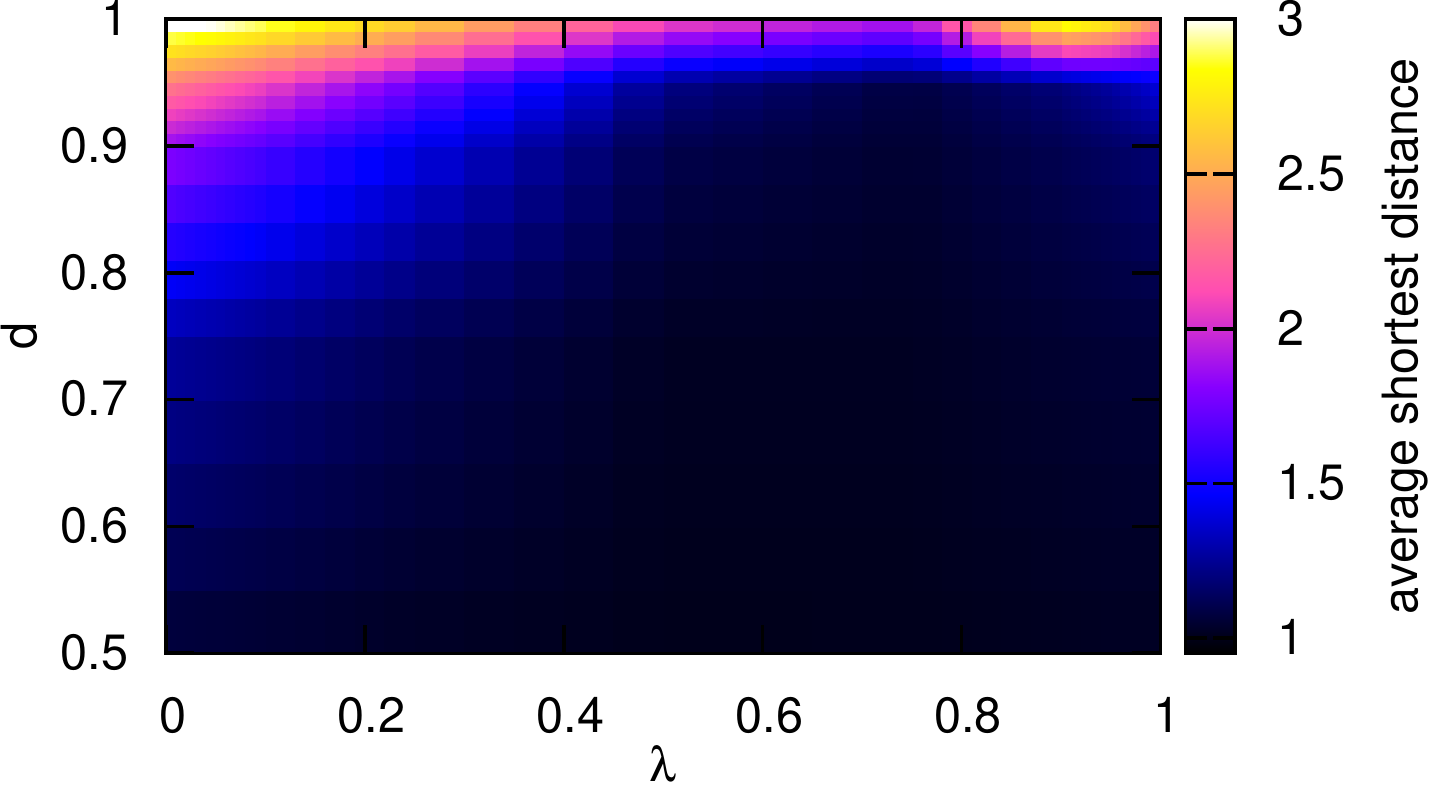}
\caption{Average shortest distance of top-10 recommendations by $\WRWR$
  from seed papers based on the parameters $d$ and $\lambda$.}
\label{fig:levels}
\end{figure}

\subsection{Dataset collection}
\label{sec:data}

The retrieval of bibliographic information and citation graph generation 
is a difficult task since academic papers are generally copyrighted and 
they are accessible through publishers' digital libraries. The usage
of such data is usually not explicitly granted, therefore, we limited
our study to data with license compatible with data mining.

We retrieved informations about 1.75M~(as of Dec 2011) computer science 
articles from DBLP~\cite{Ley09}. This data is well-formatted, author names 
are disambiguated; however, it does not contain any reference information.
On the other hand, CiteSeer
contains reference information but most of its data are automatically
generated~\cite{Giles98} and are often erroneous. We mapped each
document in CiteSeer to at most one document in DBLP by using the
title information~(using an inverted index on title words and
Levenshtein distance) and by their years. When two documents in
CiteSeer map to the same document in DBLP, their citation information
are merged. From the 1,748,199 documents references in DBLP, only
295,317 are properly associated with a reference in CiteSeer written
by 1,028,288 authors. The graph has 1,601,067 citation edges. Notice
that a mapping between CiteSeer data and DBLP data has been computed
before using canopy clustering with three times higher
coverage~\cite{Pham10}. Although we could not match a that much of the
data, we believe the data are enough to derive meaningful conclusions.

\subsection{Citation recommendation experiments}
\label{sec:expcr}

\subsubsection{Parameter tests}

Before performing a comparison of the different methods presented in
the paper, we study the impact of the damping factor $d$ and the
direction awareness parameter $\lambda$ on the recommendations given
by the \WRWR algorithm. In particular, we want to verify that changing
these parameters allows the user to obtain suggestions that are
farther away from the seed papers $\M$ and to obtain suggestions that
are either recent or more traditional. To verify these effects, a
source paper published between 2005 and 2010 is randomly selected and
the paper's references are used as the seed papers. We use the top-10
results as the set of suggestions. The test is repeated 500 times.

Figure~\ref{fig:levels} shows the impacts of parameters $d$ and
$\lambda$ as a heat map on the average shortest distance in the
citation graph between the recommended papers $\R_{paper}$ and the
seed papers $\M$. When $d$ increases, the probability that the
random research jumps back to the source node $s$ is
reduced. Therefore, the distant vertices are visited with more
probability between two successive restarts, resulting in papers away
from $\M$ being more likely to be in
$\R_{paper}$. Figure~\ref{fig:levels} shows that $\lambda$ makes
little difference in the average distance to the seed papers. However,
setting a higher value of $d$ should allow to find relevant papers
whose relation to the seeds are not obvious.

\begin{figure}[t]
\centering
\includegraphics[width=.9\linewidth]{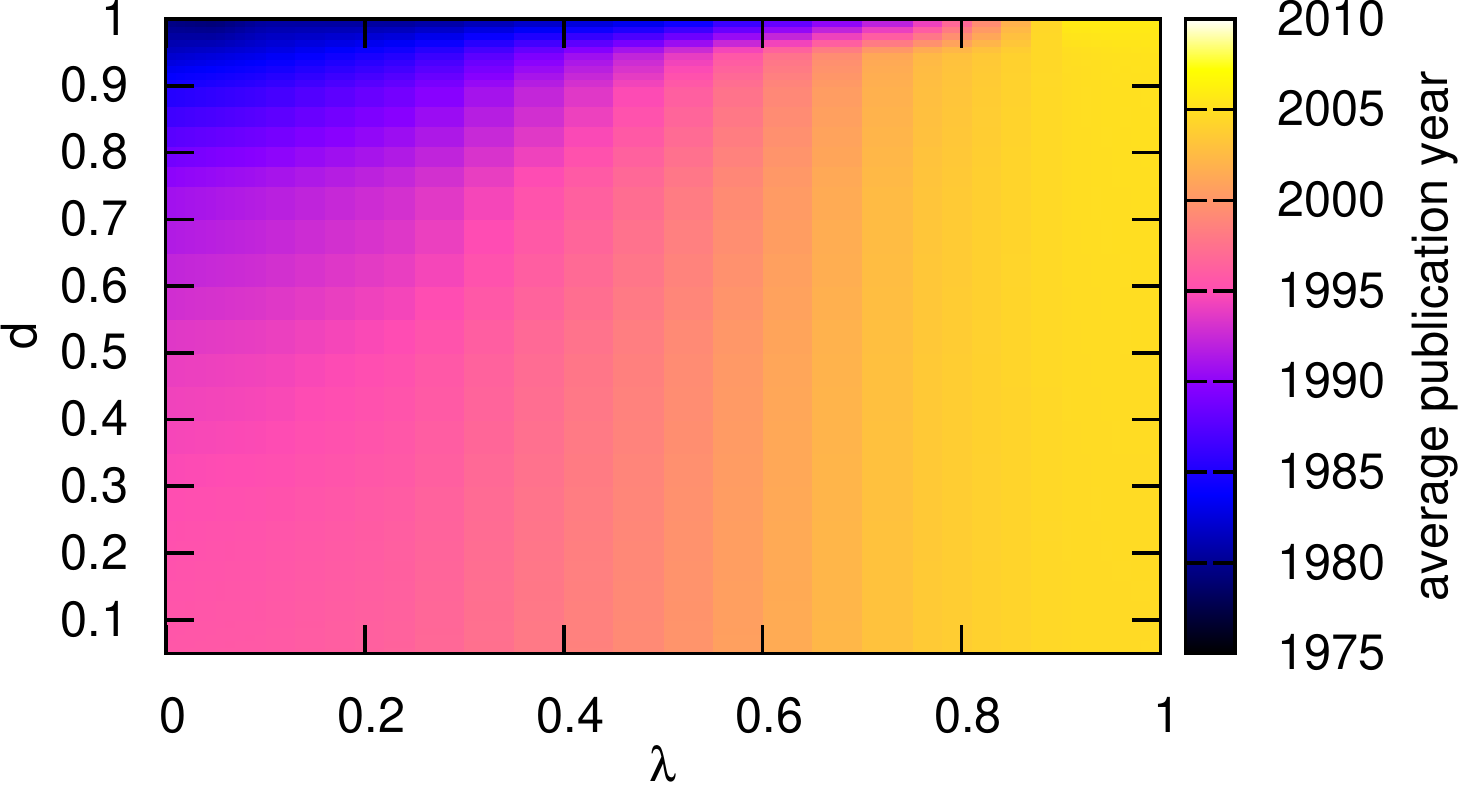}
\caption{Average publication years of top-10 recommendations by
  $\WRWR$ based on the parameters $d$ and $\lambda$.}
\label{fig:years}
\end{figure}

Figure~\ref{fig:years} shows the impacts of parameter $d$ and
$\lambda$ on the average year of the recommended papers in
$\R_{paper}$ as a heat map. Increasing the damping factor leads to
earlier papers since they tend to accumulate more citations. But for a
given $\lambda$, varying the damping factor do not allow to reach a
large diversity of time frames. The direction awareness parameter
$\lambda$ can be adjusted to reach papers from different years with a
range from late 1980's to 2010 for almost all values of $d$. In our
online service, the parameter $\lambda$ can be set to a value of
user's preference. It allows the user to obtain recent papers by
setting $\lambda$ close to $1$ or finding older papers by setting
$\lambda$ close to $0$.

Overall, first-level papers are often returned for $d<0.8$; yet many
papers at distance 2 and more appear. Also, it is possible to choose
between traditional papers~(by setting $\lambda<0.4$) or recent 
papers~(by setting $\lambda>0.8$) thanks to the direction awareness
parameter.

\subsubsection{Experimental settings}

\begin{table}[t]
\caption{Parameters used in the experiments.}
\addtolength{\tabcolsep}{-2pt}
\begin{center}
\begin{tabular}{|l|l|l|l|l|}
\hline
{\bf Method} & {\bf Random} & {\bf Recent} & {\bf Earlier} & {\bf Future}\\ \hline
Katz$_\beta$ & \multicolumn{4}{c|}{$\beta=0.0005$}\\ \hline
\multirow{2}{*}{$\WKatz$} & $\beta\!=\!0.005$ & $\beta$=$0.0005$ & $\beta$=$0.0005$ & $\beta$=$0.005$\\
~ & $\lambda=0.25$ & $\lambda=0.75$ & $\lambda=0$ & $\lambda=0.25$ \\ \hline
$\RWR$ & $d=0.5$ & $d=0.9$ & $d=0.9$ & $d=0.75$\\ \hline
\multirow{2}{*}{$\WRWR$} & $\lambda=0.5$ & $\lambda=0.9$ & $\lambda=0.1$ & $\lambda=0.5$ \\
~ & $d=0.75$ & $d=0.5$ & $d=0.5$ & $d=0.75$ \\ \hline
\end{tabular}
\end{center}
\label{tab:algos}
\end{table}%

\begin{figure*}[t]
\subfigure[hide random]{
\includegraphics[height=4.2cm,page=1,trim=0 0 2cm 0,clip=true]{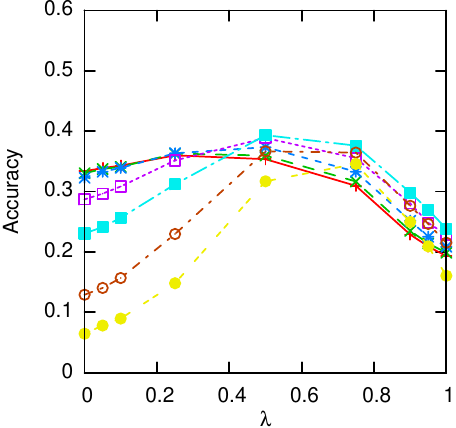}}
\subfigure[hide recent]{
\includegraphics[height=4.2cm,page=1,trim=0.8cm 0 2cm 0,clip=true]{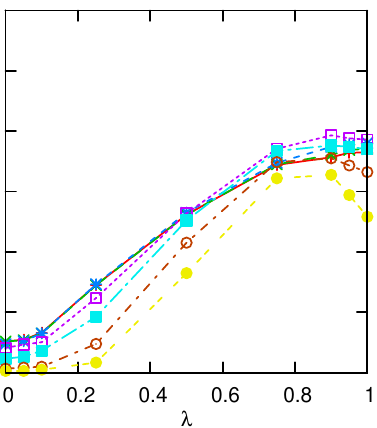}}
\subfigure[hide earlier]{
\includegraphics[height=4.2cm,page=1,trim=0.8cm 0 2cm 0,clip=true]{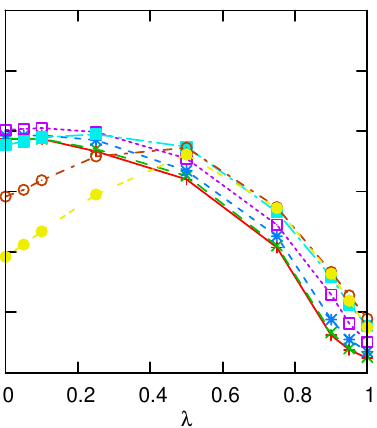}}
\subfigure[future prediction]{
\includegraphics[height=4.2cm,page=1,trim=0.8cm 0 2cm 0,clip=true]{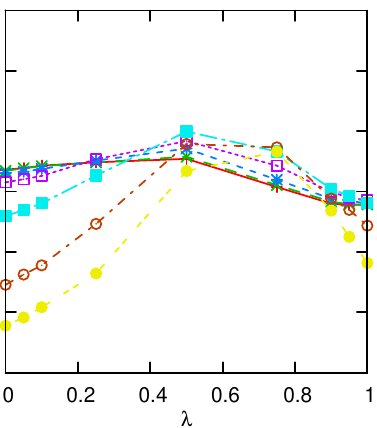}}
\subfigure{
\includegraphics[height=4.2cm,page=1,trim=5.4cm 0 0 0,clip=true]{fig9.pdf}}
\vspace{-1em}
\caption{Accuracy of $\WRWR$ method with different $\lambda$ and $d$ 
parameters on different experiments.}
\label{fig:expRWR}
\end{figure*}

\begin{figure*}[t]
\includegraphics[height=5.5cm,page=2,trim=-1.3cm 0 2.5cm 0,clip=true]{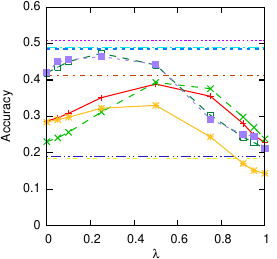}
\includegraphics[height=5.5cm,page=2,trim=0.3cm 0 0 0,clip=true]{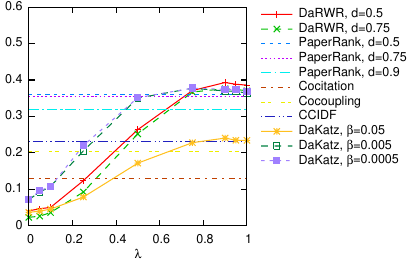}\\
\includegraphics[height=5.5cm,page=2,trim=-1.3cm 0 2.5cm 0,clip=true]{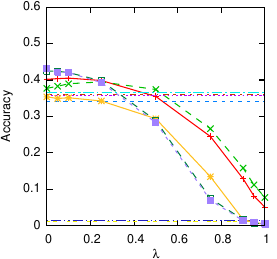}
\includegraphics[height=5.5cm,page=2,trim=0.3cm 0 2.5cm 0,clip=true]{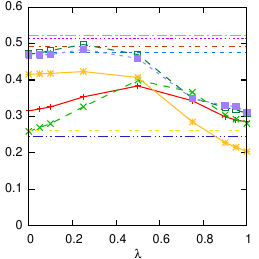}
\vspace{-1em}
\caption{Accuracy of the algorithms on (top left) hide random, (top right) hide recent, 
(bottom left) hide earlier, and (bottom right) future prediction experiments based on $\lambda$ 
and other parameters. Note that the accuracy of \Katz is equal to \WKatz
at $\lambda=0.5$.}
\label{fig:expALL}
\end{figure*}

We test the quality of the recommended citations by different
methods in four different scenarios.

{\bf Hide random} scenario represents the typical use-case
  where a researcher is writing a paper and trying to find some more
  references. To simulate that, a source paper $s$ with enough
  references ($deg^+(s) \ge 20$) is randomly selected from the
  papers published between 2005 and 2010. Then 
  we remove $s$ and all the papers published after $s$ from the
  graph (i.e., $G_s = (V_s,E_s)$ where $V_s \subset V \setminus \{s\}$
  and $\forall v \in V_s , year[v]\le year[s]$), simulating the time
  when $s$ was being written. Out of $deg^+(s)$, $10\%$ of the
  references are randomly put in the hidden set $H$, and the rest is
  used as the seed papers (i.e., $\M=\{v \notin H:(s,v) \in E\}$). We
  compute the citation recommendations on $\M$ and report the average
  accuracy of finding hidden papers within the top $deg^+(s)$
  recommendations for 500 independent queries.

{\bf Hide recent} scenario represents another typical use-case 
where the author might be well aware of the literature of
  her field but might have missed some recent developments. It
  differs from {\bf hide random} while hiding the references. Here,
  the references that are put in $H$ are not chosen randomly. They are
  the most recent references. Again, the average accuracy of finding
  hidden papers within the top $deg^+(s)$ recommendations is reported
  for each source $s$.
  
In the {\bf hide earlier} scenario, the author is interested in finding
  some key papers related to the field. This scenario is exactly the
  opposite of {\bf hide recent}, i.e., the hidden papers are the
  oldest publications. The average accuracy of finding those hidden
  traditional papers within the top $deg^+(s)$ recommendations is
  reported for each source $s$.  

{\bf Future prediction} scenario investigates the accuracy of a recommendation system while providing a link between two papers which
  are not known to be related yet. It verifies if the
  algorithm can predict which paper will be cited by a given
  paper. For this test, the source paper $s$ is selected
  similarly. However, the graph selected for the recommendation
  include paper $s$ but exclude all subsequent papers (i.e., $G_s =
  (V_s,E_s)$ with $v \in V_s \iff year[v] \le year[s]$). And all the
  references of the $s$ are used as the seeds to obtain a top-10
  recommendations. The accuracy of the algorithm is estimated by
  counting how many of the documents that appear in the top-10 is
  later co-cited with the source paper.

The methods we proposed are compared on the three scenarios against
widely-used citation based approaches: bibliographic
coupling~\cite{Kessler63}, Cocitation~\cite{Small73},
CCIDF~\cite{Lawrence99}, \RWR~\cite{Gori06} and the original Katz
distance~\cite{Liben-Nowell03}. The algorithms and the parameters that
lead to the best accuracy in different experiments are summarized in
Table~\ref{tab:algos}.

\subsubsection{Results}

Figure~\ref{fig:expRWR} presents the accuracy obtained by the \WRWR
for different combinations of the parameters $d$ and $\lambda$ on the
four scenarios. The results show that extreme values of the parameter
are typically not the one that obtain the highest accuracy. On the
hide random experiment, \WRWR performs best with $d=0.75$ and $\lambda
= 0.5$. A similar combination set~($d=0.75, \lambda=0.9$) obtains a
high accuracy on the hide recent experiment. However it is best
processed with parameters $d=0.5$ and $\lambda=0.9$. As expected, the
hide earlier experiment is best solved using a low value of the
direction awareness parameter~($d=0.5, \lambda=0.1$). The future
prediction experiment is best solved by the $d=0.75, \lambda=0.5$
parameter set. Still using $d=0.5$ leads to solutions of reasonable
accuracy. It is interesting to notice that the hide random and future
prediction experiments show similar pattern while the hide recent and
hide earlier experiments show opposite patterns. This experiment tells
us that it is enough to set $\lambda$ as tunable for the service since
tuning $d$ has little impact once it is set to a reasonable
value. Most likely, setting $d$ as tunable will add only more
complexity and no significant improvement in the accuracy.

Figure~\ref{fig:expALL} presents a comparison of all the methods on
the same scenarios. Many algorithms are represented as horizontal
lines since they are not direction aware. The first remark is that
Cocoupling and CCIDF perform poorly on all four scenarios. Cocitation
performs the worse in the hide recent scenario and performs reasonably good
but not the best in the other three scenarios. These methods which only
consider counting and weighting of distance 2 edges at most from the
seeds are out-performed by the eigenvector based methods which take whole
graph into account.

Notice that \RWR performs well overall but for different values of the
damping parameter $d$. The performance of \WKatz is significantly
varying with the parameter set but it is important to notice that the
variations with the direction awareness parameter are similar to the
one observed on \WRWR. The results of \Katz are not explicitly
presented but can be read on \WKatz when $\lambda=0.5$. Notice that
\WKatz is always a better method that \Katz. \RWR achives the best
results when the query is generic~(on the hide random and future
prediction scenarios); however direction aware methods lead to higher
accuracy when the query is specific.

The previous experiments show that the method we proposed return
results of higher accuracy. However, these results do not allow us to
understand whether the methods return similar results or different
results. Table~\ref{tab:cm} presents the intersection matrix of the
different methods on four scenarios. Each method's parameters are set
to optimize the accuracy. The diagonal of the matrix shows the actual
accuracy of the methods. Other values show the percentage of the
intersection of two corresponding methods. For instance, one can read
that on the hide random scenario, \RWR has an accuracy of $51.30\%$
while CCIDF has an accuracy of $20.12\%$. The intersection between the
results of CCIDF and \RWR has an accuracy of $17.23\%$ indicating that
most of the relevant results returned by CCIDF were also results by
\RWR in that scenario. In the hide recent and hide random scenarios,
the proposed method clearly dominate the solution space. The other
methods do not add many new relevant suggestions.

\begin{table*}[t]
\caption{Intersection matrix of the results for (i) hide random, 
(ii) hide recent, (iii) hide earlier, and (iv) future prediction 
experiments.}  \scriptsize
\begin{center}
\addtolength{\tabcolsep}{-2.9pt}
\begin{minipage}{1.05 \columnwidth}
\begin{tabular}{@{}l|rrrrrrr@{}}
\multicolumn{1}{c|}{(i)} & $\WRWR$ & \textsc{P.R.} & $\WKatz$ & Katz$_\beta$ & Cocit & Cocoup & CCIDF \\ \hline
$\WRWR$ 		& \bf{40.62}&	38.96&	36.75&	33.81&	30.02&	13.48&	14.46\\
\textsc{P.R.} 	& ~	& \bf{51.31}&	43.59&	40.31&	35.18&	16.20&	17.23\\
$\WKatz$		& ~	& ~	& \bf{48.72}&	39.63&	35.99&	15.57&	15.31\\
Katz$_\beta$ 	& ~	& ~	& ~ & \bf{44.87}&	31.10&	17.17&	16.89\\
Cocit 			& ~	& ~	& ~ & ~	& \bf{42.57}&	11.53&	11.00\\
Cocoup			& ~	& ~	& ~ & ~	& ~	& \bf{19.47}&	15.04\\
CCIDF  			& ~	& ~	& ~ & ~	& ~	& ~	& \bf{20.13}\\
\multicolumn{8}{c}{}%
\end{tabular}
\end{minipage}%
\begin{minipage}{1.05 \columnwidth}
\begin{tabular}{@{}l|rrrrrrr@{}}
\multicolumn{1}{c|}{(ii)} & $\WRWR$ & \textsc{P.R.} & $\WKatz$ & Katz$_\beta$ & Cocit & Cocoup & CCIDF \\ \hline
$\WRWR$ 		& \bf{40.57}&	33.51&	31.68&	31.13&	7.86&	16.78&	19.92\\
\textsc{P.R.} 	& ~	& \bf{37.41}&	30.89&	31.37&	9.67&	17.19&	20.18\\
$\WKatz$		& ~	& ~	& \bf{38.18}&	35.72&	8.48&	19.28&	21.19\\
Katz$_\beta$ 	& ~	& ~	& ~ & \bf{37.18}&	9.35&	19.07&	21.06\\
Cocit 			& ~	& ~	& ~ & ~	& \bf{13.87}&	6.28&	5.96\\
Cocoup			& ~	& ~	& ~ & ~	& ~	& \bf{22.03}&	18.08\\
CCIDF  			& ~	& ~	& ~ & ~	& ~	& ~	&\bf{25.23}\\
\multicolumn{8}{c}{}%
\end{tabular}
\end{minipage}

\begin{minipage}{1.05 \columnwidth}
\begin{tabular}{@{}l|rrrrrrr@{}}
\multicolumn{1}{c|}{(iii)} & $\WRWR$ & \textsc{P.R.} & $\WKatz$ & Katz$_\beta$ & Cocit & Cocoup & CCIDF \\ \hline
$\WRWR$ 		& \bf{60.72}&	51.21&	56.92&	41.28&	46.61&	1.97&	2.35\\
\textsc{P.R.} 	& ~	& \bf{55.17}&	52.73&	40.39&	45.94&	1.88&	2.29\\
$\WKatz$		& ~	& ~	& \bf{65.11}&	42.69&	50.67&	2.21&	2.44\\
Katz$_\beta$ 	& ~	& ~	& ~ & \bf{43.04}&	39.53&	2.10&	2.35\\
Cocit 			& ~	& ~	& ~ & ~	& \bf{53.02}&	1.95&	2.09\\
Cocoup			& ~	& ~	& ~ & ~	& ~	& \bf{2.48}&	1.18\\
CCIDF  			& ~	& ~	& ~ & ~	& ~	& ~	&\bf{2.81}\\
\end{tabular}
\end{minipage}%
\begin{minipage}{1.05 \columnwidth}
\begin{tabular}{@{}l|rrrrrrr@{}}
\multicolumn{1}{c|}{(iv)} & $\WRWR$ & \textsc{P.R.} & $\WKatz$ & Katz$_\beta$ & Cocit & Cocoup & CCIDF \\ \hline
$\WRWR$ 		& \bf{39.08}&	28.75&	24.59&	20.82&	18.91&	5.68&	6.31\\
\textsc{P.R.} 	& ~	& \bf{51.48}&	32.55&	30.87&	24.50&	9.56&	10.57\\
$\WKatz$		& ~	& ~	& \bf{49.37}&	26.50&	30.66&	6.34&	5.21\\
Katz$_\beta$ 	& ~	& ~	& ~ & \bf{45.15}&	17.41&	13.99&	12.30\\
Cocit 			& ~	& ~	& ~ & ~	& \bf{48.65}&	3.41&	2.48\\
Cocoup			& ~	& ~	& ~ & ~	& ~	& \bf{25.22}&	14.78\\
CCIDF  			& ~	& ~	& ~ & ~	& ~	& ~	&\bf{24.27}\\
\end{tabular}
\end{minipage}%
\end{center}
\label{tab:cm}
\end{table*}%

The case of the future prediction scenario is different. The
intersection between the different methods often highlight that a
significant portion of the returned suggestion differ between the
algorithms. For instance, the intersection between \WRWR and
Cocoupling scores an accuracy of $5.68\%$ which is 5 times smaller
than the accuracy of Cocoupling ($25.22\%$) and 7.5 times smaller than
the accurary of \WRWR ($39.08\%$).

\subsubsection{Citation patterns}

For a better understanding of the difference between the accuracy obtain by
different methods, we did a study on the properties of the suggestions
returned by the methods and compare them to the properties of the
actual references within the papers. We argue that highly relevant
suggested papers should have similar patterns to the actual
references.

One feature to measure the citation patterns is the {\it clustering
  coefficient}~\cite{Watts98}. The clustering coefficient $C_v$ of
paper $v$ is computed as:
\[
C_v = \frac{|\{(i,j)\in E \mid i,j \in N_v \union \{v\}\}|}{|N_v| \times (|N_v| + 1)}, 
\]
where $N_v$ is the set of neighbor papers of $v$ which either cite $v$
or are cited by $v$. Intuitively, the clustering coefficient indicates
how close of being a clique a vertex and its neighbors are.

The other metric we consider is the PageRank~\cite{Brin98} of a vertex
which can be calculated by putting all vertices in \M during the
\RWR algorithm.

Figure~\ref{fig:ccpr} presents the cumulative density function of the
clustering coefficient and of the PageRank of the documents suggested
by each algorithm and of the hidden papers in the three hidden 
scenarios. The first observation is that on all charts the Cocitation
algorithm is an outlier. Also, CCIDF and Cocoupling are almost
indistinguishable on all charts. Interestingly, the clustering
coefficient of the hidden papers in the hide earlier scenario are
lower than in the hide random scenario and the clustering coefficient
of the hidden paper in the hidden recent scenario are the highest. The
trend is reverse with PageRank. Older papers have more time to become
famous so their PageRank is higher. And since they have more citations,
it is less likely that their neighbors are close to form a clique. This
highlights that papers published in different years have different
profiles, bolstering our claim that one should not use the very same
algorithm and parameters to look for them.

\begin{figure*}[t]
\centering
\includegraphics[height=5.5cm,page=3]{fig10.pdf}
\includegraphics[height=5.5cm,page=1,trim=.8cm 0 0 0,clip=true]{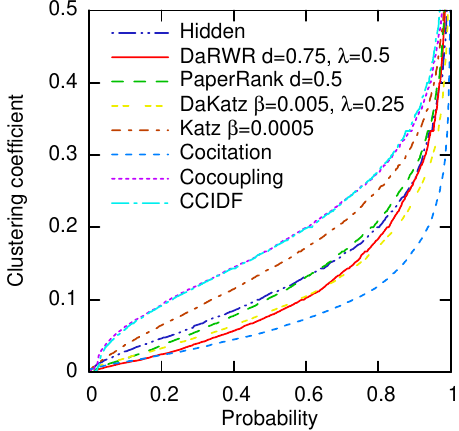}
\includegraphics[height=5.5cm,page=2,trim=.8cm 0 0 0,clip=true]{fig10.pdf}\\
\vspace{1.9em}
\includegraphics[height=5.5cm,page=3]{fig11.pdf}
\includegraphics[height=5.5cm,page=1,trim=1.2cm 0 0 0,clip=true]{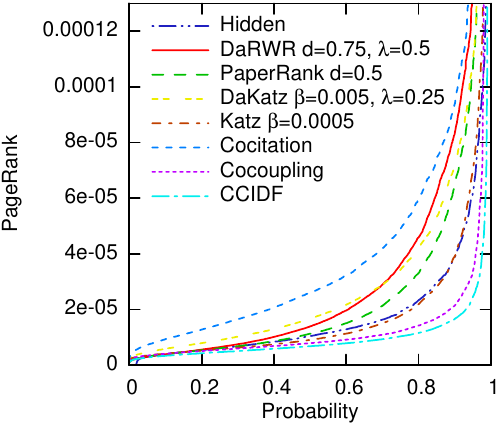}
\includegraphics[height=5.5cm,page=2,trim=1.2cm 0 0 0,clip=true]{fig11.pdf}
\caption{Clustering coefficient (top) and Pagerank (bottom) of the suggested 
citations for the hide earlier (left), hide random (center), and hide recent
(right) experiments.}
\label{fig:ccpr}
\end{figure*}

For the hide random scenario, \RWR suggests papers of clustering
coefficient very similar to the hidden paper, while \WRWR, \Katz, and
\WKatz show a different but parallel trace. The PageRank distribution
of the algorithm shows a similar picture, except \Katz is close the
hidden paper and \RWR, \WRWR, and \WKatz are farther away.

In the hide recent scenario, most algorithms have a similar trace for
both the clustering coefficient and PageRank. \RWR and \Katz are
significantly different than their direction aware variants and the
trace of the hidden paper. Recall that \RWR and \Katz are also less
accurate than their direction aware variants on the hide recent
scenario. Having a similar trace is an important property but it is
not enough to reach a high accuracy. Indeed, Cocoupling and CCIDF
show a trace similar to the that of hidden papers in that scenario but
with less accuracy.

In the hide earlier scenario, the direction aware algorithm have
patterns similar to the hidden paper for both metric explaining the
high accuracy they reach. \RWR has a PageRank pattern similar to the
hidden paper but a different clustering coefficient pattern and it does not
reach the high accuracy level the direction aware algorithms
obtain. \Katz's pattern is similar to that of the hidden paper neither in
clustering coefficient nor on PageRank and it is the one with the lowest 
accuracy among all the eigenvector based methods.

This analysis shows that direction aware algorithms have overall
similar citation patterns. CCIDF and cocoupling have typically similar
citation patterns. The difference in accuracy of the eigenvector based
methods can be explained by the similarity in citation patterns between
the papers one is looking for and what is generated by the method. The
direction aware methods are more flexible and can be tuned to match
the property of the query leading to higher accuracy. The reasons of
success or failure of the non-eigenvector based methods (Cocitation,
Cocoupling, and CCIDF) seem to be unrelated to the citation pattern
metrics we considered.

\subsection{Relevance feedback experiments}
\label{sec:exprelfeed}

Relevance feedback is an important part of the recommendation system
since users may give positive and negative feedbacks on the results in
order to reach to desired papers or topics. In this test, 500 source
papers are randomly selected, and for each source paper $s$ the graph
is pruned by removing the papers published after $s$. Then, a target
paper $u$ is selected from the pruned graph, such that it is the most
relevant paper at distance $5$ from $u$. Assuming that a user can only
display 10 results at a time, we measure the number of pages that the
user has to go through until she reaches $t$. We compare the feedback
mechanism with the following idealized user behavior:

\begin{description}[leftmargin=1.2em,itemsep=-0.2em]
\vspace{-0.3em}
\item[{\bf No feedback:}] There is no feedback mechanism; therefore, user 
should keep looking the next page until she finds the target paper.
\item[{\bf Only positive feedback:}] Results are labeled as relevant and
  added to $\M$ in the next step or should not be displayed again.
\item[{\bf Only negative feedback:}] Results are labeled 
as irrelevant to be removed from the graph or should not be displayed again.
\item[{\bf Both positive and negative:}] Results are labeled as either 
relevant to be added to $\M$ or irrelevant to be removed from the graph.
\end{description}

Detailed results for that experiment are omitted. Using negative
feedback only reduces the number of pages one has to go through by
82.29\% in average and using positive feedback allows to reduce the
number of pages by 97.15\% in average. Using both negative and
positive feedback reduces the number of pages by 97.20\% in
average. This result shows that using the feedback mechanism allows to
significantly speedup the process of searching for specific references.

\subsection{Venue and reviewer recommendation experiments}
\label{sec:expvr}

The venue recommendation methods is tested on the assumption that a
paper is published in a venue where it is relevant. The following
protocol relies on this assumption. A source paper is randomly
selected and is removed from the graph as long as all subsequent
papers. The objective is to find the venue of the source paper in
$\R_{venue}$ containing $k = 10$ venues. We compare the performance of
our methods against a method commonly employed by researcher, which
consist in considering the top-10 most occurring venues of the paper
of interest; e.g., the \M set. We call this algorithm~{\bf Baseline~1}. 
Another algorithm, {\bf Baseline~2}, considers the venues of the paper at
distance 2 of the source paper: it returns the top-10 most occurring
venues in \M and the references and citation of these documents.

The reviewer recommendation experiment is based on the assumption that
``the authors are the best reviewers for the paper" (ignoring the
obvious \emph{conflict-of-interest}, and by best reviewers referring
to people that have the enough knowledge on this candidate paper). The
experiment is conducted similarly to the venue recommendation
experiment. A source paper is selected and is removed from the graph
as long as all subsequent papers. For a list $\R_{expert}$ which
contains $k = 25$ experts, we distinguish whether none of the authors
of the source paper is found, if any author is found or if all the
authors are found. Both baselines are defined in the same way as in the
venue recommendation experiment.

\begin{table}
\caption{Average accuracy of venue recommendation (VR) and reviewer 
recommendation (RR) experiments.}
\begin{center}
\begin{tabular}{l|c|c|c}
~ & {\bf VR} & \multicolumn{2}{c}{\bf RR}\\ \hline
~ & Accuracy@10 & Any@25 & All@25\\ \hline
$\WRWR$ & 63.2		& 76.4 & 48.19\\
$\RWR$ & 60.6		& 74.4 & 45.85\\
$\WKatz$ & 58.4		& 64.4 & 35.17\\ \hline %
Baseline~1 & 56.0	& 73.0 & 48.38\\
Baseline~2 & 60.0	& 72.6 & 44.04
\end{tabular}
\end{center}
\label{tab:vrrr}
\end{table}%

Table~\ref{tab:vrrr} presents the average accuracy of these methods
when run on 500 random (uniform) source papers. For venue
recommendation, the three proposed methods perform better than
Baseline~1 and \WRWR perform better than Baseline~2. The differences
are marginal (less than 10\%) but statistically significant. For
reviewer recommendation, \WRWR performs the best. Interestingly Baseline~2 
performs worse than Baseline~1 in both experiments.

\section{Conclusion and future work}
\label{sec:conclusion}

In this paper, we present direction aware algorithms for citation
recommendation. A direction aware model allows to tune the search for
finding more recent or more traditional documents. We developed two
algorithms based on the direction aware model, namely \WKatz and
\WRWR. We also suggest to use the classical random walk with restart
(\RWR) for academic recommendation. Experimentally, we confirmed that
the parameters can be easily set to browse the academic web of
knowledge. In our experiments, the direction aware algorithm we
propose outperforms the existing algorithms for citation
recommendation which are based only on the citation graph in
experiments that focus on finding either traditional or recent
papers. We implemented the algorithms in our webservice which
allows any researcher to upload a bibliography file and obtain
suggestions. This service is freely available and easy to use. Coupled
with our efficient algorithms, we believe that our service will become
a tool of major interest for researchers.

As future work, we want to improve our service both in theory and
practice. We are planning to test weighting schemes on edges to have a
better distribution of probability to papers with high quality. In
practice, we will improve the amount and the quality of the
bibliographic data by using existing techniques such as canopy
clustering and by obtaining data from more public academic databases. 
We are also planning to conduct an intensive user study to
obtain a real-world evaluation of the system.

\section*{Acknowledgments}
This work was partially supported by the U.S. Department of Energy 
SciDAC Grant DE-FC02-06ER2775 and NSF grants CNS-0643969, OCI-0904809 
and OCI-0904802.
The authors also would like to thank DBLP and CiteSeer for
making their data publicly available.

\bibliographystyle{abbrv}

\small
\bibliography{report}

\balancecolumns

\end{document}